\documentclass[11pt, a4paper,oneside]{article}
\usepackage{amsthm}
\usepackage{amsmath}
\usepackage{amssymb}
\usepackage{secdot}
\usepackage{fontawesome}
\usepackage{cite}
\usepackage{tabularx}
\usepackage{rotating}
\usepackage{enumitem}
\usepackage[dvipsnames]{xcolor}
\usepackage{setspace}
\usepackage[left=1in, right=1in, top=1.2in, bottom=1.2in]{geometry}
\usepackage{titlesec}
\usepackage{xpatch}
\usepackage{ctable}
\usepackage[hang,flushmargin]{footmisc} 
\newtheorem{theorem}{Theorem}[section]

\newtheorem{lemma}[theorem]{Lemma}
\newtheorem{definition}{Definition}[section]
\newtheorem{example}[definition]{Example}
\newtheorem{proposition}[definition]{Proposition}
\newtheorem{remark}[definition]{Remark}
\renewenvironment{proof}{{\noindent\bf{Proof }}}{\hfill $\blacksquare$\\}
\titleformat{\section}
{\normalfont\large\bfseries}{\thesection}{0.4em}{}
\titleformat{\subsection}
{\normalfont\bfseries}{\thesubsection}{0.4em}{}
\makeatletter
\AtBeginDocument{\xpatchcmd{\@thm}{\thm@headpunct{.}}{\thm@headpunct{}}{}{}}
\makeatother
\allowdisplaybreaks
\begin{document}	
	\noindent \rule{\textwidth}{0.5mm}\\
\vspace{0.1cm}\\
\noindent\textbf{\Large Five-Lee-weight linear codes over $\mathbb{F}_{q}+u\mathbb{F}_{q}$}\\
\vspace{0.1cm}\\
\textbf{Pavan Kumar$^{\boldsymbol{1,2},\ast}$ and Noor Mohammad Khan$^{\boldsymbol{1}}$}\\
 {$^{1}$Department of Mathematics, Aligarh Muslim University, Aligarh 202002, India}\\
 {$^{2}$Department of Electronic Systems Engineering, Indian Institute of Science, Bengaluru 560012, India}\\
\let\thefootnote\relax\footnote{\noindent$^{\ast}$Corresponding author\\  Pavan Kumar\\ pavan4957@gmail.com\\  \vspace*{0.1cm}\\ Noor Mohammad Khan\\nm\_khan123@yahoo.co.in}
\vspace{1.5cm}\\
\textbf{\large Abstract}\\
In this study, linear codes having their Lee-weight distributions over the semi-local ring $\mathbb{F}_{q}+u\mathbb{F}_{q}$ with $u^{2}=1$ are constructed using the defining set and Gauss sums for an odd prime $q $. Moreover, we derive complete Hamming-weight enumerators for the images of the constructed linear codes under the Gray map. We finally show an application to secret sharing schemes.

\section{Introduction}
  The study of linear codes with few weights has contributed significantly to various fields such as communications, data storage, and consumer electronics, leading to significant advancements in their study over finite fields. Researchers have focused on developing linear codes with few weights specifically for secret sharing \cite{ADHK98, CDY05}, association schemes \cite{CG84}, and authentication codes \cite{DW05}.

Throughout the paper, for some positive integer $m$ and an odd prime $q$, let $\mathbb{F}_{q^{m}}$ be a finite field having $q^{m}$ elements. A subspace of $\mathbb{F}_{q}^{n}$ is referred to as a linear code $\mathcal{C}$ over the field $\mathbb{F}_{q}$, where $n$ is the length of the code $\mathcal{C}$. A linear code $\mathcal{C}$ over $\mathbb{F}_{q}$ having length $n$, dimension $k$ and minimum Hamming-distance $d$ is usually known as $[n,k,d]$ linear code. The \emph{weight enumerator} of the code $\mathcal{C}$ is a polynomial defined by
$$1+f_{1}x+f_{2}x^{2}+\cdots+f_{n}x^{n},$$
where $f_{i}~(1\leq i\leq n)$ represents the number of codewords of weight $i$ in $\mathcal{C}$. Moreover, the $(n+1)$-tuple $(1,f_{1},f_{2},\ldots,f_{n})$ is known as the weight distribution of the code $\mathcal{C}$. In coding theory, the weight enumerator of a linear code $\mathcal{C}$ is an important topic, as it provides insights into error-detection and error-correction probabilities, as well as the error-correcting capability with respect to various algorithms \cite{DD15}. Throughout this paper, $\#S$ denotes the cardinality of any set $S$. A code $\mathcal{C}$ satisfying the equation $$\#\{i:f_{i}\neq0, 1\leq i\leq n\}=t$$ is known as $t$-weight code.

 Let $\mathbb{F}_{q}=\{w_{0},w_{1},\ldots, w_{q-1}\}$, where $w_{0}$ denotes the additive identity of the field. Let $x=(x_{0}, x_{1},\ldots, x_{n-1})$ be any vector in $\mathbb{F}_{q}^{n}$. Then the composition of $x$, denoted by comp($x$), is represented by $q$-tuple $(n_{0}, n_{1},\ldots, n_{q-1}),$
where $n_{j}$ is defined by $$n_{j}=\#\{x_{i}:x_{i}=w_{j}(0\leq i\leq n-1)\}.$$ Let $A(n_{0},n_{1},\ldots,n_{q-1})=\#\{c\in\mathcal{C}:\text{com}(c)=(n_{0},n_{1},\ldots,n_{q-1})\}$ and $\mathbb{P}_{n}=\{(n_{0},n_{1},\ldots,n_{q-1}):0\leq n_{j}\leq n,\sum_{j=0}^{q-1}n_{j}=n\}$. The complete weight enumerator of the code $\mathcal{C}$ which demonstrates the frequency of each symbol appearing in each codeword  is the polynomial defined  as follows:
\begin{align*}
\text{CWE($\mathcal{C}$)}&=\sum_{c\in\mathcal{C}}w_{0}^{n_{0}}w_{1}^{n_{1}}\cdots w_{q-1}^{n_{q-1}}\\
&=\sum_{(n_{0},n_{1},\ldots,n_{q-1})\in\mathbb{P}_{n}} A(n_{0},n_{1},\ldots,n_{q-1})w_{0}^{n_{0}}w_{1}^{n_{1}}\cdots w_{q-1}^{n_{q-1}}.
\end{align*}
The study of linear codes with their complete weight enumerators has found numerous applications in communication technologies, leading to extensive research in this area \cite{KK21,LBAYY16,AKL17,YY17,YY172,WGF17,LL18}. These weight enumerators play a crucial role in calculating the deception probability of authentication codes, as demonstrated by Ding et al. in \cite{DHKW07,DW05}. 

Let $\mathbb{F}_{q^{m}}^{*}=\mathbb{F}_{q^{m}}\setminus \{0\}$ and $D=\{d_{1},d_{2},\ldots,d_{n}\}\subset\mathbb{F}_{q^{m}}^{*}$, and let $\mathrm{Tr}$ denote the absolute trace function from $\mathbb{F}_{q^{m}}$ onto $\mathbb{F}_{q}$. The set $\mathcal{C}_{D}$ defined by the following equation \eqref{eq1} forms a linear code of length $n$ over $\mathbb{F}_{q}$ 
\begin{equation}\label{eq1}
\mathcal{C}_{D}=\{(\mathrm{Tr}(xd_{1}),\mathrm{Tr}(xd_{2}),\ldots,\mathrm{Tr}(xd_{n})):x\in\mathbb{F}_{q^{m}}\},
\end{equation}
where the set $D$ is known as the defining set of this code $\mathcal{C}_{D}$. First ever such codes were constructed by Ding et al. \cite{DD15}. Applicable linear codes can be effectively constructed by strategically choosing the defining set, as discussed in \cite{DD15,KK24,KK21}.

In the 1990s, modules over finite rings became influential in coding theory, impacting diverse fields such as number theory (unimodular lattices \cite{A98}), engineering (low correlation sequences \cite{BGO07}), and combinatorics (designs \cite{AM74}). Nechaev \cite{N89} and Hammons et al. \cite{HKCSS94} made significant contributions by establishing the connection between linear codes over the ring \(\mathbb{Z}_4\) and nonlinear binary codes like the Kerdock \cite{K17} and Preparata \cite{P68} codes. Codes over finite rings typically possess more codewords compared to those over finite fields of the same length. For instance, the Preparata code, despite having the same length as the extended double error-correcting BCH code, contains twice as many codewords as the BCH code \cite{HKCSS94}. This illustrates the richness and potential of codes over finite rings in expanding the capacity and efficiency of coding systems.
 
Let  $\mathcal{R}=\mathbb{F}_{q}+u\mathbb{F}_{q}$ be a ring, where $u$ is an indeterminant such that $\mathbb{F}_{q}+u\mathbb{F}_{q}$ forms a ring under usual addition and multiplication. An $\mathcal{R}$-submodule of $\mathcal{R}^{n}$ is called a linear code $\mathcal{C}$ of length $n$ over $\mathcal{R}=\mathbb{F}_{q}+u\mathbb{F}_{q}$. Assume that $\mathcal{R}_{m}=\mathbb{F}_{q^{m}}+u\mathbb{F}_{q^{m}}$ be a ring extension of $\mathcal{R}~(=\mathbb{F}_{q}+u\mathbb{F}_{q})$ with degree $m$.  Based on approach in \eqref{eq1}, one can find linear codes with their weight distributions over the ring $\mathcal{R}$, but in this case $D=\{d_{1},d_{2},\ldots,d_{n}\}\subset\mathcal{R}_{m}$, and $\text{tr}:\mathcal{R}_{m} \rightarrow \mathcal{R}$ is a linear function  \cite{P68,SLS16,LL19, ZZG24, GZMF23}.

The Lee-weight distributions of linear codes over the finite non-chain ring $\mathbb{F}_{q}+u\mathbb{F}_{q}$, with $u^{2}=1$, are examined in this study. Also, we have determined the complete Hamming-weight enumerators of the images of constructed linear codes under the Gray map. The ring and the defining set that we have considered to construct codes are more general than that  in \cite{LL19}. Moreover, in Remark \ref{re1}, we have shown that our code has improved relative minimum distance compared to a code in \cite{LL19}.

The complete study is organized as follows: In Section \ref{Sec.2}, we provide an overview of definitions and results concerning group characters and Gauss sums over finite fields. Additionally, we present several propositions regarding the ring $\mathbb{F}_{q} + u\mathbb{F}_{q}$. Section \ref{Sec.3} contains auxiliary results that are crucial for proving the main results. The main results themselves are established in Section \ref{Sec.4}. Section \ref{Sec.5} explores an application of these results in a secret sharing scheme. Finally, Section \ref{Sec.6} summarizes our findings and concludes the work.

 \section{Preliminaries}\label{Sec.2}
$\mathbb{F}_{q}+u\mathbb{F}_{q}$ be a semi-local ring with the natural addition and multiplication with the property $u^{2}=1$. The ring $\mathbb{F}_{q}+u\mathbb{F}_{q}$, with $u^{2}=1$,  has the following two maximal ideals 
\begin{equation*}
I_{1}=\langle u-1 \rangle=\{(r+us)(u-1):r,s\in\mathbb{F}_{q}\};
\end{equation*}
\begin{equation*}
	I_{2}=\langle u+1 \rangle=\{(r+us)(u+1):r,s\in\mathbb{F}_{q}\}.
\end{equation*}
 It can easily be seen that both $\mathcal{R}/I_{1}$ and  $\mathcal{R}/I_{2}$ are isomorphic to $\mathbb{F}_{q.}$
Let $q$ be an odd prime and $m\geq 2$ be positive integer. We construct the ring extension $\mathcal{R}_{m}=\mathbb{F}_{q^{m}}+u\mathbb{F}_{q^{m}}$ of $\mathcal{R}~(=\mathbb{F}_{q}+u\mathbb{F}_{q})$ with degree $m$. The \emph{Trace function} from $\mathbb{F}_{q^{m}}+u\mathbb{F}_{q^{m}}$ onto $\mathbb{F}_{q}+u\mathbb{F}_{q}$ is defined as
$$\mathrm{tr}=\sum_{j=o}^{m-1}F^{j},$$
where $F$ is the operator which maps $r+us$ to $r^{q}+us^{q}$ and known as Frobenius operator.
For all $r, s\in\mathbb{F}_{q^{m}}$, it can easily be checked that
$$\mathrm{tr}(r+us)=\mathrm{Tr}(r)+u\mathrm{Tr}(s),$$  where $\mathrm{Tr}$ stands for absolute trace function from $\mathbb{F}_{q^{m}}$ onto $\mathbb{F}_{q}$. The one to one map $\phi$ from $\mathcal{R}$ to $\mathbb{F}_{q}^{2}$ defined by
$$ \phi(r+us)=(r,s),\quad\text{for all } r, s\in\mathbb{F}_{q}^{n},$$ is known as Gray map.
This map can naturally be extended into a map from $\mathcal{R}^{n}$ to $\mathbb{F}_{q}^{2n}$. The \emph{Lee-weight} of $r'+us'\in \mathcal{R}$ is defined as the \emph{Hammimg-weight} of its Gray image i.e,
$$w_{L}(r'+us')=w_{H}(r')+w_{H}(s'),\quad\text{for all } r', s'\in\mathbb{F}_{q}^{n}.$$
Let $x, y\in \mathcal{R}^{n}$.  The \emph{Lee-distance} between $x$ and $y$, which is denoted by $d_{L}(x,y)$, is the \emph{Lee-weight} of $(x-y)$ i.e., $d_{L}(x,y)=w_{H}(x-y)$. So, one can easily check that the Gray map is, by construction, a linear isometry from $(\mathcal{R}^{n}, d_{L}) $ to $(\mathbb{F}_{q}^{2n}, d_{H})$.

We now give certain group character and Gauss sum results for usage in the future (for further information, see \cite{LN97}).

Let $S^{1}=\{z\in\mathbb{C}: |z|=1\}$, where $\mathbb{C}$ stands for the set of complex numbers. It is easy to verify that $S^{1}$ is a group under usual multiplication of complex numbers. A mapping $\chi: \mathbb{F}_{q^{m}}\rightarrow S^{1}$ defined by  $\chi(r_{1} + r_{2})=\chi(r_{1})\chi(r_{2})$ for all  $ r_{1},  r_{2}\in \mathbb{F}_{q^{m}} $ is called an additive character  of $\mathbb{F}_{q^{m}}$ \cite{LN97}.

Let $ a\in  \mathbb{F}_{q^{m}} $. Then, by Theorem 5.7 in \cite{LN97},
\begin{equation}\label{2.1}
\chi_{a}(r)= \zeta_{q}^{\mathrm{Tr}(ar)},\quad \text{for  all}~r\in \mathbb{F}_{q^{m}},
\end{equation}
defines an additive character of $ \mathbb{F}_{q^{m}} $, where $\zeta_{q}=e^{\frac{2\pi\sqrt{-1}}{q}}$, and $\{\chi_{a}\}_{a\in \mathbb{F}_{q^{m}}}$ is the only collection of all additive characters which can be defined on  $\mathbb{F}_{q^{m}}$. A trivial character, denoted by $\chi_{0}$,  is an additive character whose range contains only the identity of $S^{1}$ i.e., $ \chi_{0}(r)=1 $ for all $ r\in \mathbb{F}_{q^{m}}$. Whereas, the characters $\chi_{a}$, where $a\neq0$, are known as nontrivial characters. The canonical additive character of $ \mathbb{F}_{q^{m}}$  is obtained by replacing $a$ by $1$ in \eqref{2.1}\cite{LN97}.\\
By Theorem 5.4 in \cite{LN97}, the orthogonal property of additive character of $\mathbb{F}_{q^{m}}$ is defined as 
\begin{equation}
\sum_{r\in \mathbb{F}_{q^{m}}}\chi(r)=
\begin{cases}
q^{m},&\text{if $\chi$ trivial},\\
0,&\text{if $\chi$ non-trivial}.
\end{cases}
\end{equation}
The multiplicative character of $ \mathbb{F}_{q^{m}} $ is the character of its multiplicative group $ \mathbb{F}^{*}_{q^{m}}$. Let $g$ be a generator of $ \mathbb{F}^{*}_{q^{m}} $. Then by Theorem 5.8 in \cite{LN97}, for each $i=0, 1, \ldots, q^{m}-2$, the function $ \psi_{i} $ with
\begin{equation*}
	\psi_{i}(g^{j})=e^{\frac{2\pi\sqrt{-1}ij}{q^{m}-1}} \quad\text{for}\  j=0, 1,\ldots, q^{m}-2
\end{equation*}
defines a multiplicative character of $ \mathbb{F}_{q^{m}} $. For $ i=\frac{q^{m}-1}{2} $, we have the quadratic character $ \eta=\psi_{\frac{q^{m}-1}{2}} $ defined by
$$\eta(g^{k})=
\begin{cases}
-1, &\text{if  $ 2\nmid k$},\\
1,&\text{if $ 2\mid k$}. 
\end{cases}
$$
In the sequel, we shall assume that $ \eta(0)=0 $. Quadratic Gauss sum over $ \mathbb{F}_{q^{m}} $, denoted by $ G=G(\eta, \chi_{1})$, is defined as 
\begin{equation*}
	G(\eta, \chi_{1})=\sum_{r\in \mathbb{F}_{q^{m}}^{*}}\eta(r)\chi_{1}(r),
\end{equation*}
where $\chi_{1} $ and $\eta$  denote the canonical  and quadratic characters of $ \mathbb{F}_{q^{m}} $, respectively. Similarly, quadratic Gauss sum over $ \mathbb{F}_{q} $, denoted by
 $ \overline{G}=G(\overline{\eta}, \overline{\chi}_{1}) $, is defined as 
 \begin{equation*}
 	G(\overline{\eta}, \overline{\chi_{1}})=\sum_{r'\in \mathbb{F}_{q}^{*}}\overline{\eta}(r')\overline{\chi}_{1}(r'), 
 \end{equation*}
where $ \overline{\chi}_{1} $ and $ \overline{\eta} $ denote the canonical and quadratic characters of $ \mathbb{F}_{q} $, respectively.
The following lemma gives the explicit quadratic Gauss sum values.
\begin{lemma}\label{p3le1}
	{\em \cite[Theorem 5.15]{LN97}} Let the symbols have the same meanings as before. Then
	$$G(\eta, \chi_{1})=(-1)^{m-1}\sqrt{-1}^{\frac{(q-1)^{2}m}{4}}\sqrt{q^{m}},~ G(\overline{\eta}, \overline{\chi}_{1})=\sqrt{-1}^{\frac{(q-1)^{2}}{4}}\sqrt{q}.$$
\end{lemma}
\begin{lemma}\label{p3le2}
	{\em \cite[Lemma 2.2]{LL19}} Let the symbols continue to mean what they did before. When $m\geq 2$ is even, $\eta(y)=1 $ for each $y$ in $\mathbb{F}^{*}_{q}$, and when $m \geq 2$ is odd, $\eta(y)=\overline{\eta}(y) $ for each $y\in \mathbb{F}^{*}_{q} $.
	\end{lemma}
\begin{lemma}\label{p3le3}
	{\em\cite[Theorem 5.33]{LN97}} Let $ \chi $ be a non trivial additive character of $ \mathbb{F}_{q^{m}} $, and let $f(x)=b_{2}x^{2} + b_{1}x + b_{0}\in \mathbb{F}_{q^{m}}[x]$ with $ b_{2}\neq 0 $. Then
	$$\sum_{x\in \mathbb{F}_{q^{m}}}\chi(f(x))=\chi(b_{0} - b_{1}^{2}(4b_{2})^{-1})\eta(b_{2})G(\eta, \chi).$$
\end{lemma}
\begin{lemma}\label{p3le4}
{\em \cite[Lemma 9]{DD15}} For $s\in\mathbb{F}_{q}$,   let $N_{s}=\# \{\alpha\in\mathbb{F}_{q^{m}} : \mathrm{Tr}(\alpha^{2})=s\}.$ Then
$$N_{s}=\begin{cases}
q^{m-1}+q^{-1}(q-1)G,&\text{if $s=0$ and $2\mid m$},\\
q^{m-1},&\text{if $s=0$ and $2\nmid m$},\\
q^{m-1}+q^{-1}\overline{\eta}(-s)G\overline{G},&\text{if $s\neq0$ and $2\nmid m$},\\
q^{m-1}-q^{-1}G,&\text{if $s\neq0$ and $2\mid m$}.
\end{cases} $$
\end{lemma}
\begin{lemma}\label{p3le5}
	{\em \cite[Proof of Lemma 9]{DD15}} Let $ s\in\mathbb{F}_{q}$. Then
	$$\sum_{x\in\mathbb{F}_{q}^{*}}\sum_{a\in\mathbb{F}_{q^{m}}}\zeta_{q}^{x(\mathrm{Tr}(a^{2})-s)}=\begin{cases}
	(q-1)G,&\text{if $s=0$ and $2\mid m$},\\
	0,&\text{if $s=0$ and $2\nmid m$},\\
	\overline{\eta}(-s)G\overline{G},&\text{if $s\neq0$ and $2\nmid m$},\\
	-G,&\text{if $s\neq0$ and $2\mid m$}.
	\end{cases} $$
\end{lemma}
\section{Auxiliary Results}\label{Sec.3}
In this section, we present various lemmas before providing the main results.
\begin{lemma}\label{p3le7}
	For $s, t\in\mathbb{F}_{q}$, let
	$$N(s,t)=\#\{(\alpha, \beta)\in\mathbb{F}_{q^{m}}\times\mathbb{F}_{q^{m}} : \mathrm{Tr}(\alpha^{2})=s~and~\mathrm{Tr}(\beta^{2})=t\}.$$
	Then, for odd $m$, we have
	$$N(s, t)=\begin{cases}
		q^{2m-2},&\text{if $s=0$ and $t=0$},\\
		q^{2m-2}+q^{m-2}\overline{\eta}(-s)G\overline{G},&\text{if $s\neq0$ and $t=0$},\\
		q^{2m-2}+q^{m-2}\overline{\eta}(-t)G\overline{G},&\text{if $s=0$ and $t\neq0$},\\
		q^{2m-2}+q^{m-2}\big(\overline{\eta}(-s)+\overline{\eta}(-t)\big)G\overline{G}+q^{-2}\overline{\eta}(st)G^{2}\overline{G}^{2},&\text{if $s\neq0$ and $t\neq0$};
	\end{cases}$$
	and, for even $m$, we have
	$$N(s, t)=\begin{cases}
		(q^{m-1}+q^{-1}(q-1)G)^{2},&\text{if $s=0$ and $t=0$},\\
		q^{2m-2}+q^{m-2}(q-2)G-q^{-2}(q-1)G^{2},&\text{if $s\neq0$ and $t=0$},\\
		q^{2m-2}+q^{m-2}(q-2)G-q^{-2}(q-1)G^{2},&\text{if $s=0$ and $t\neq0$},\\
		(q^{m-1}-q^{-1}G)^{2},&\text{if $s\neq0$ and $t\neq0$}.
	\end{cases}~~~~~~~~~~~~~~$$
\end{lemma}
\begin{proof} The proof is directly follows from the Lemma \ref{p3le4}.\end{proof}
\begin{lemma}\label{p3le9}
	For $\beta\in\mathbb{F}_{q}^{*}$ and $\lambda\in\mathbb{F}_{q}^{*},$
	let
	$$N_{1}=\sum_{a\in\mathbb{F}_{q^{m}}}\sum_{x\in\mathbb{F}_{q}^{*}}\zeta_{q}^{x\mathrm{Tr}(a^{2})}\sum_{z\in\mathbb{F}_{q}^{*}}\zeta_{q}^{z\mathrm{Tr}(\beta a)-z\lambda}.$$ Then
	$$N_{1}=\begin{cases}
		-G(q-1), & \text{if }2\mid m~\text{and}~\mathrm{Tr}(\beta^{2})=0,\\
		G, & \text{if }2\mid m~\text{and}~\mathrm{Tr}(\beta^{2})\neq0,\\
		0, & \text{if }2\nmid m~\text{and}~\mathrm{Tr}(\beta^{2})=0,\\
		-\overline{\eta}(-\mathrm{Tr}(\beta^{2}))G\overline{G}, & \text{if }2\nmid m~\text{and}~\mathrm{Tr}(\beta^{2})\neq0.\end{cases}$$
\end{lemma}
\begin{proof} Lemmas \ref{p3le2} and \ref{p3le3} give us
	\begin{align*}
		N_{1}
		&=\sum_{x,z\in\mathbb{F}_{q}^{*}}\overline{\chi}_{1}(-z\lambda)\sum_{a\in\mathbb{F}_{q^{m}}}\chi_{1}(xa^{2}+z\beta a)\\&=
		G\sum_{x,z\in\mathbb{F}_{q}^{*}}\overline{\chi}_{1}(-z\lambda)\chi_{1}(-\frac{z^{2}\beta^{2}}{4x})\eta(x)\\
		&=\begin{cases}
			G\sum\limits_{x,z\in\mathbb{F}_{q}^{*}}\overline{\chi}_{1}\big(-\frac{z^{2}\mathrm{Tr}(\beta^{2})}{4x}-z\lambda\big), & \text{if }2\mid m\\
			G\sum\limits_{x,z\in\mathbb{F}_{q}^{*}}\overline{\chi}_{1}\big(-\frac{z^{2}\mathrm{Tr}(\beta^{2})}{4x}-z\lambda\big)\overline{\eta}(x), & \text{if }2\nmid m
		\end{cases}\allowdisplaybreaks\\
		&=
		\begin{cases}
			G\sum\limits_{x,z\in\mathbb{F}_{q}^{*}}\overline{\chi}_{1}(-z\lambda), & \text{if }2\mid m~\text{and}~\mathrm{Tr}(\beta^{2})=0\\
			G\sum\limits_{x,z\in\mathbb{F}_{q}^{*}}\overline{\chi}_{1}\big(-\frac{z^{2}\mathrm{Tr}(\beta^{2})}{4x}-z\lambda\big), & \text{if }2\mid m~\text{and}~\mathrm{Tr}(\beta^{2})\neq0\\
			G\sum\limits_{x,z\in\mathbb{F}_{q}^{*}}\overline{\chi}_{1}(-z\lambda)\overline{\eta}(x), & \text{if }2\nmid m~\text{and}~\mathrm{Tr}(\beta^{2})=0\\
			G\sum\limits_{x,z\in\mathbb{F}_{q}^{*}}\overline{\chi}_{1}\big(-\frac{z^{2}\mathrm{Tr}(\beta^{2})}{4x}-z\lambda\big)\overline{\eta}(x), & \text{if }2\nmid m~\text{and}~\mathrm{Tr}(\beta^{2})\neq0
		\end{cases}\\
		&=\begin{cases}
			-G(q-1), & \text{if }2\mid m~\text{and}~\mathrm{Tr}(\beta^{2})=0,\\
			G, & \text{if }2\mid m~\text{and}~\mathrm{Tr}(\beta^{2})\neq0,\\
			0, & \text{if }2\nmid m~\text{and}~\mathrm{Tr}(\beta^{2})=0,\\
			-\overline{\eta}(-\mathrm{Tr}(\beta^{2}))G\overline{G}, & \text{if }2\nmid m~\text{and}~\mathrm{Tr}(\beta^{2})\neq0.
		\end{cases}
	\end{align*}	
	The proof is completed.\end{proof}
\begin{lemma}\label{p3le10}
	For $\beta\in\mathbb{F}_{q}^{*}$ and $\lambda\in\mathbb{F}_{q}^{*}$, let
	$$N_{2}=\sum_{a,b\in\mathbb{F}_{q^{m}}}\sum_{x\in\mathbb{F}_{q}^{*}}\zeta_{q}^{x\mathrm{Tr}(a^{2})}\sum_{y\in\mathbb{F}_{q}^{*}}\zeta_{q}^{y\mathrm{Tr}(b^{2})}\sum_{z\in\mathbb{F}_{q}^{*}}\zeta_{q}^{z\mathrm{Tr}(\beta a)-z\lambda}.$$ Then
	$$N_{2}=\begin{cases}
		-G^{2}(q-1)^{2},&\text{if }2\mid m~and~\mathrm{Tr}(\beta^{2})=0, \\
		G^{2}(q-1),&\text{if }2\mid m~and~\mathrm{Tr}(\beta^{2})\neq0,\\
		0,&\text{if }2\nmid m.
	\end{cases}$$
\end{lemma}
\begin{proof} 
	One may easily prove the lemma by using the similar arguments that were used to prove the previous lemma.
\end{proof}

\begin{lemma}\label{p3le11}
	For $\alpha,\beta\in\mathbb{F}																																																																																																																																																																																																										    _{q^{m}}^{*}$ and $\lambda\in\mathbb{F}_{q}^{*}$, let
	$$N_{3}=\sum_{a,b\in\mathbb{F}_{q^{m}}}\sum_{x\in\mathbb{F}_{q}^{*}}\zeta_{q}^{x\mathrm{Tr}(a^{2})}\sum_{y\in\mathbb{F}_{q}^{*}}\zeta_{q}^{y\mathrm{Tr}(b^{2})}\sum_{z\in\mathbb{F}_{q}^{*}}\zeta_{q}^{z\mathrm{Tr}(\alpha b+\beta a)-z\lambda}.$$ Then	
	$$N_{3}=\begin{cases}
		-G^{2}(q-1)^{2},& \text{if }2\mid m,~ \mathrm{Tr}(\alpha^{2})=0\text{ and } \mathrm{Tr}(\beta^{2})=0,\\
		G^{2}(q-1),& \text{if }2\mid m,~ \mathrm{Tr}(\alpha^{2})=0\text{ and } \mathrm{Tr}(\beta^{2})\neq0,\\
		G^{2}(q-1),& \text{if }2\mid m,~ \mathrm{Tr}(\alpha^{2})\neq0\text{ and } \mathrm{Tr}(\beta^{2})=0,\\
		-G^{2},& \text{if }2\mid m,~ \mathrm{Tr}(\alpha^{2})\neq0\text{ and } \mathrm{Tr}(\beta^{2})\neq0,\\
		-\overline{\eta}\big(\mathrm{Tr}(\alpha^{2})\mathrm{Tr}(\beta^{2})\big)G^{2}\overline{G}^{2},&\text{if }2\nmid m,~\mathrm{Tr}(\alpha^{2})\neq0~\text{and}~\mathrm{Tr}(\beta^{2})\neq0,\\
		0,&\text{if }2\nmid m,~\mathrm{Tr}(\alpha^{2})=0~\text{or}~\mathrm{Tr}(\beta^{2})=0.
	\end{cases}	$$
\end{lemma}
\begin{proof} From Lemmas \ref{p3le2}  and \ref{p3le3}, we have
	\begin{align*}
		N_{3}&=\sum_{x,z\in\mathbb{F}_{q}^{*}}\overline{\chi}_{1}(-z\lambda)\sum_{a\in\mathbb{F}_{q^{m}}}\chi_{1}(xa^{2}+z\beta a)\sum_{y\in\mathbb{F}_{q}^{*}}\sum_{b\in\mathbb{F}_{q^{m}}}\chi_{1}(yb^{2}+z\alpha b)\\
		&=\sum_{x,z\in\mathbb{F}_{q}^{*}}\overline{\chi}_{1}(-z\lambda)\chi_{1}(-\frac{z^{2}\beta^{2}}{4x})\eta(x)G\sum_{y\in\mathbb{F}_{q}^{*}}\chi_{1}(-\frac{z^{2}\alpha^{2}}{4y})\eta(y)G\\
		&=G^{2}\sum_{x,z\in\mathbb{F}_{q}^{*}}\overline{\chi}_{1}(-z\lambda)\chi_{1}(-\frac{z^{2}\beta^{2}}{4x})\eta(x)\sum_{y\in\mathbb{F}_{q}^{*}}\chi_{1}(-\frac{z^{2}\alpha^{2}}{4y})\eta(y)\\
		&=G^{2}\sum\limits_{x,z\in\mathbb{F}_{q}^{*}}\overline{\chi_{1}}(-\frac{z^{2}\mathrm{Tr}(\beta^{2})}{4x}-z\lambda)\eta(x)\sum\limits_{y\in\mathbb{F}_{q}^{*}}\overline{\chi_{1}}(-\frac{z^{2}\mathrm{Tr}(\alpha^{2})}{4y})\eta(y).
	\end{align*}
	Hereafter, we divide the proof into two parts depending on whether 2 divides $m$ or not.\\
	\textbf{Case 1:} Consider $2\mid m$, then we have
	\begin{align*}
		N_{3}&=\begin{cases}		G^{2}(q-1)\sum\limits_{x,z\in\mathbb{F}_{q}^{*}}\overline{\chi_{1}}(-z\lambda),&\text{if }\mathrm{Tr}(\alpha^{2})=0~\text{and}~\mathrm{Tr}(\beta^{2})=0\\
			G^{2}(q-1)\sum\limits_{x,z\in\mathbb{F}_{q}^{*}}\overline{\chi_{1}}(-\frac{z^{2}\mathrm{Tr}(\beta^{2})}{4x}-z\lambda),&\text{if }\mathrm{Tr}(\alpha^{2})=0~\text{and}~\mathrm{Tr}(\beta^{2})\neq0\\	-G^{2}\sum\limits_{x,z\in\mathbb{F}_{q}^{*}}\overline{\chi_{1}}(-z\lambda),&\text{if }\mathrm{Tr}(\alpha^{2})\neq0 \text{ and } \mathrm{Tr}(\beta^{2})=0\\
			-G^{2}\sum\limits_{x,z\in\mathbb{F}_{q}^{*}}\overline{\chi_{1}}(-\frac{z^{2}\mathrm{Tr}(\beta^{2})}{4x}-z\lambda),&\text{if }\mathrm{Tr}(\alpha^{2})\neq0 \text{ and } \mathrm{Tr}(\beta^{2})\neq0
		\end{cases}\\
		&=\begin{cases}		-G^{2}(q-1)^{2},&\text{if }\mathrm{Tr}(\alpha^{2})=0~\text{and}~\mathrm{Tr}(\beta^{2})=0,\\
			G^{2}(q-1),&\text{if }\mathrm{Tr}(\alpha^{2})=0~\text{and}~\mathrm{Tr}(\beta^{2})\neq0,\\	
			G^{2}(q-1),&\text{if }\mathrm{Tr}(\alpha^{2})\neq0 \text{ and } \mathrm{Tr}(\beta^{2})=0,\\
			-G^{2},&\text{if }\mathrm{Tr}(\alpha^{2})\neq0 \text{ and } \mathrm{Tr}(\beta^{2})\neq0.
		\end{cases}
	\end{align*}
	\textbf{Case 2:} Suppose $2\nmid m$, then we have
	\begin{align*}
		&=\begin{cases}
			G^{2}\sum\limits_{x,z\in\mathbb{F}_{q}^{*}}\overline{\chi_{1}}(-z\lambda)\overline{\eta}(x)\sum\limits_{y\in\mathbb{F}_{q}^{*}}\overline{\eta}(y),&\text{if }\mathrm{Tr}(\alpha^{2})=0 \text{ and } \mathrm{Tr}(\beta^{2})=0\\
			G^{2}\sum\limits_{x,z\in\mathbb{F}_{q}^{*}}\overline{\chi_{1}}(-\frac{z^{2}\mathrm{Tr}(\beta^{2})}{4x}-z\lambda)\overline{\eta}(x)\sum\limits_{y\in\mathbb{F}_{q}^{*}}\overline{\eta}(y),&\text{if }\mathrm{Tr}(\alpha^{2})=0 \text{ and } \mathrm{Tr}(\beta^{2})\neq0\\
			G^{2}\sum\limits_{x,z\in\mathbb{F}_{q}^{*}}\overline{\chi_{1}}(-z\lambda)\overline{\eta}(x)\sum\limits_{y\in\mathbb{F}_{q}^{*}}\overline{\chi_{1}}(-\frac{z^{2}\mathrm{Tr}(\alpha^{2})}{4y})\overline{\eta}(y),&\text{if }\mathrm{Tr}(\alpha^{2})\neq0 \text{ and } \mathrm{Tr}(\beta^{2})=0\\
			G^{2}\sum\limits_{x,z\in\mathbb{F}_{q}^{*}}\overline{\chi_{1}}(-\frac{z^{2}\mathrm{Tr}(\beta^{2})}{4x}-z\lambda)\overline{\eta}(x)\sum\limits_{y\in\mathbb{F}_{q}^{*}}\overline{\chi_{1}}(-\frac{z^{2}\mathrm{Tr}(\alpha^{2})}{4y})\overline{\eta}(y),&\text{if }\mathrm{Tr}(\alpha^{2})\neq0 \text{ and } \mathrm{Tr}(\beta^{2})\neq0
		\end{cases}\\
		&=\begin{cases}
			0,&\text{if }\mathrm{Tr}(\alpha^{2})=0 \text{ or } \mathrm{Tr}(\beta^{2})=0,\\
			-\overline{\eta}\big(\mathrm{Tr}(\alpha^{2})\mathrm{Tr}(\beta^{2})\big)G^{2}\overline{G}^{2},&\text{if }\mathrm{Tr}(\alpha^{2})\neq0 \text{ and } \mathrm{Tr}(\beta^{2})\neq0.
		\end{cases}
	\end{align*}
	By combining the above two cases, we get the required proof.\end{proof}

\begin{lemma}\label{p3le13}
	For $\alpha,\beta\in\mathbb{F}_{q^{m}}$ and $\lambda\in\mathbb{F}_{q}^{*}$, let
	$$\Omega_{\lambda}(\alpha,\beta)=\#\{(a, b)\in\mathbb{F}_{q^{m}}\times\mathbb{F}_{q^{m}} : \mathrm{Tr}(a^{2})=0,~\mathrm{Tr}(b^{2})=0~and~\mathrm{Tr}(\alpha b+\beta a)=\lambda\}.$$
	Then\\
	$1$. if $\alpha=0$ and $\beta=0$, we have $\Omega_{\lambda}(0,0)=0;$\\
	$2$. if $\alpha=0$ and $\beta\neq0$, we have
	$$\Omega_{\lambda}(0,\beta)=\begin{cases}
		q^{2m-3}+G(q-1)q^{m-3}, & \text{ if } 2\mid m ~and~\mathrm{Tr}(\beta^{2})=0,\\
		q^{2m-3}+G(2q-1)q^{m-3}+G^{2}(q-1)q^{-2}, & \text{ if }2\mid m ~and~\mathrm{Tr}(\beta^{2})\neq0,\\
		q^{2m-3}, & \text{ if }2\nmid m ~and~\mathrm{Tr}(\beta^{2})=0,\\
		q^{2m-3}+\overline{\eta}(-\mathrm{Tr}(\beta^{2}))G\overline{G}q^{m-3}, & \text{ if }2\nmid m ~and~\mathrm{Tr}(\beta^{2})\neq0;
	\end{cases}~~~~~~~$$
	$3$. if $\alpha\neq0$ and $\beta=0$, we have
	$\Omega_{\lambda}(\alpha,\beta)=\Omega_{\lambda}(\alpha,0)=\Omega_{\lambda}(0,\alpha);$\\
	$4$. if $\alpha\neq0$ and $\beta\neq0$, we have
	{\small$$\Omega_{\lambda}(\alpha,\beta)=\begin{cases}
			q^{2m-3}+2G(q-1)q^{m-3},& \text{ if } 2\mid m,~ \mathrm{Tr}(\alpha^{2})=0\text{ and }\mathrm{Tr}(\beta^{2})=0,\\
			q^{2m-3}+2G(q-1)q^{m-3}+G^{2}(q-1)q^{-2},& \text{ if } 2\mid m,~ \mathrm{Tr}(\alpha^{2})=0\text{ and } \mathrm{Tr}(\beta^{2})\neq0,\\
			q^{2m-3}+2G(q-1)q^{m-3}+G^{2}(q-1)q^{-2},&\text{ if }2\mid m,~ \mathrm{Tr}(\alpha^{2})\neq0\text{ and } \mathrm{Tr}(\beta^{2})=0,\\
			q^{2m-3}+2G(q-1)q^{m-3}+G^{2}(q-2)q^{-2},&\text{ if } 2\mid m,~ \mathrm{Tr}(\alpha^{2})\neq0\text{ and }\mathrm{Tr}(\beta^{2})\neq0,\\
			q^{2m-3}-\overline{\eta}\big(\mathrm{Tr}(\alpha^{2})\mathrm{Tr}(\beta^{2})\big)G^{2}\overline{G}^{2}q^{-3},& \text{ if }2\nmid m,~\mathrm{Tr}(\alpha^{2})\neq0\text{ and }\mathrm{Tr}(\beta^{2})\neq0,\\
			q^{2m-3},& \text{ if }2\nmid m,~\mathrm{Tr}(\alpha^{2})=0\text{ or } \mathrm{Tr}(\beta^{2})=0.
		\end{cases}$$}
\end{lemma}
\begin{proof} The characteristics of additive characters give us
	\begin{align*}
		\Omega_{\lambda}(\alpha,\beta)&=q^{-3}\sum_{a,b\in\mathbb{F}_{q^{m}}}\sum_{x\in\mathbb{F}_{q}}\zeta_{q}^{x\mathrm{Tr}^{m}_{1}(a^{2})}\sum_{y\in\mathbb{F}_{q}}\zeta_{q}^{y\mathrm{Tr}^{m}_{1}(b^{2})}\sum_{z\in\mathbb{F}_{q}}\zeta_{q}^{z \mathrm{Tr}^{m}_{1}(\alpha b+\beta a)-z\lambda}\\
		&=q^{-3}\sum_{a,b\in\mathbb{F}_{q^{m}}}\sum_{x\in\mathbb{F}_{q}}\zeta_{q}^{x\mathrm{Tr}^{m}_{1}(a^{2})}\sum_{y\in\mathbb{F}_{q}}\zeta_{q}^{y\mathrm{Tr}^{m}_{1}(b^{2})}\left(1+\sum_{z\in\mathbb{F}_{q}^{*}}\zeta_{q}^{z \mathrm{Tr}^{m}_{1}(\alpha b+\beta a)-z\lambda}\right)\\
		&=q^{-1}N(0,0)+q^{-3}\sum_{a,b\in\mathbb{F}_{q^{m}}}\sum_{x\in\mathbb{F}_{q}}\zeta_{q}^{x\mathrm{Tr}^{m}_{1}(a^{2})}\sum_{y\in\mathbb{F}_{q}}\zeta_{q}^{y\mathrm{Tr}^{m}_{1}(b^{2})}\sum_{z\in\mathbb{F}_{q}^{*}}\zeta_{q}^{z \mathrm{Tr}^{m}_{1}(\alpha b+\beta a)-z\lambda}\\
		&=q^{-1}N(0,0)+q^{-3}(\Omega_{1}+\Omega_{2}+\Omega_{3}+\Omega_{4}),
	\end{align*}
	where 
	\begin{align*}
		\Omega_{1}&=\sum_{a,b\in\mathbb{F}_{q^{m}}}\sum_{z\in\mathbb{F}_{q}^{*}}\zeta_{q}^{z \mathrm{Tr}^{m}_{1}(\alpha b+\beta a)-z\lambda},\\
		\Omega_{2}&=\sum_{a,b\in\mathbb{F}_{q^{m}}}\sum_{x\in\mathbb{F}_{q}^{*}}\zeta_{q}^{x\mathrm{Tr}^{m}_{1}(a^{2})}\sum_{z\in\mathbb{F}_{q}^{*}}\zeta_{q}^{z \mathrm{Tr}^{m}_{1}(\alpha b+\beta a)-z\lambda},\\
		\Omega_{3}&=\sum_{a,b\in\mathbb{F}_{q^{m}}}\sum_{y\in\mathbb{F}_{q}^{*}}\zeta_{q}^{y\mathrm{Tr}^{m}_{1}(b^{2})}\sum_{z\in\mathbb{F}_{q}^{*}}\zeta_{q}^{z \mathrm{Tr}^{m}_{1}(\alpha b+\beta a)-z\lambda},\\
		\Omega_{4}&=N_{4}=\sum_{a,b\in\mathbb{F}_{q^{m}}}\sum_{x\in\mathbb{F}_{q}^{*}}\zeta_{q}^{x\mathrm{Tr}^{m}_{1}(a^{2})}\sum_{y\in\mathbb{F}_{q}^{*}}\zeta_{q}^{y\mathrm{Tr}^{m}_{1}(b^{2})}\sum_{z\in\mathbb{F}_{q}^{*}}\zeta_{q}^{z \mathrm{Tr}^{m}_{1}(\alpha b+\beta a)-z\lambda}.
	\end{align*}		
	Thus, the result directly follows from the Lemmas \ref{p3le5}, \ref{p3le7}, \ref{p3le9}, \ref{p3le10} and \ref{p3le11}.\end{proof}

After completing the aforementioned foundational work, we are now prepared to prove our main results in the following theorems.
\section{Code Construction and Their Weight Distribution}\label{Sec.4}
In this section, we introduce our codes and subsequently provide their Lee-weight distribution and complete Hamming-weight enumerators of their Gray images, respectively.

We define the defining set $D=\{a+ub\in\mathcal{R}_{m}^{*}: \mathrm{tr}(a^{2}+ub^{2})=0\}$ and let $D=\{a+ub\in\mathcal{R}_{m}^{*}: \mathrm{tr}(a^{2}+ub^{2})=0\}$=$\{a_{1}+ub_{1},a_{2}+ub_{2},\ldots,a_{n}+ub_{n}\}\subset\mathcal{R}_{m}^{*}$. Then, for $x=\alpha+u\beta\in\mathcal{R}_{m}$, we define, based on construction in \eqref{eq1}, our code $\mathcal{C}_{D}$ as follows:
\begin{align}
	\mathcal{C}_{D}&=\{(\mathrm{tr}(xa_{1}+uxb_{1}),\ldots,\mathrm{tr}(xa_{n}+uxb_{n})):x\in\mathcal{R}_{m}\}\label{eq4}\\
	&=\{(\mathrm{Tr}(\alpha a_{1}+\beta b_{1})+u\mathrm{Tr}(\beta a_{1}+\alpha b_{1}),\ldots,\mathrm{Tr}(\alpha a_{n}+\beta b_{n})+u\mathrm{Tr}(\beta a_{n}+\alpha b_{n})):\alpha, \beta\in\mathbb{F}_{q^{m}}\}\nonumber,
\end{align}
and
$$\phi(\mathcal{C}_{D})=\{(\mathrm{Tr}(\alpha a_{1}+\beta b_{1}),\mathrm{Tr}(\beta a_{1}+\alpha b_{1}),\ldots,\mathrm{Tr}(\alpha a_{n}+\beta b_{n}),\mathrm{Tr}(\beta a_{n}+\alpha b_{n})):\alpha, \beta\in\mathbb{F}_{q^{m}}\}.$$

The following theorems present the Lee-weight distribution of $\mathcal{C}_{D}$ and the complete Hamming-weight enumerator of $\phi(\mathcal{C}_{D})$, depending on whether $2$ divides $m$ or not.
\begin{theorem}\label{p3theorem1}
	Assume that $m\geq3$ is odd. The parameters for the presented  linear code $\mathcal{C}_{D}$ are $[q^{2m-2}-1,m].$ Additionally,  Table 1 contains its Lee-weight distribution, and the complete Hamming-weight enumerator of $\phi(\mathcal{C}_{D})$ is
	\begin{align*}
		&w_{0}^{2q^{2m-2}-2}+P_{0}w_{0}^{Q_{0}-2}\prod_{i=1}^{p-1}w_{i}^{Q_{0}}+P_{-1}w_{0}^{Q_{0}+2(q-1)q^{\frac{3m-5}{2}}-2}\prod_{i=1}^{q-1}w_{i}^{Q_{0}-2q^{\frac{3m-5}{2}}}\\&+P_{1}w_{0}^{Q_{0}-2(q-1)q^{\frac{3m-5}{2}}-2}\prod_{i=1}^{q-1}w_{i}^{Q_{0}+2q^{\frac{3m-5}{2}}}+Q_{1}w_{0}^{Q_{0}+2(q-1)q^{m-2}-2}\prod_{i=1}^{q-1}w_{i}^{Q_{0}-2q^{m-2}}\\&+Q_{-1}w_{0}^{Q_{0}-2(q-1)q^{m-2}-2}\prod_{i=1}^{q-1}w_{i}^{Q_{0}+2q^{m-2}},
	\end{align*}
	where $P_{i}=(q-1)(q^{m-1}+iq^{\frac{m-1}{2}})(i=1,-1),\mbox{$Q_{j}=\frac{(q-1)2}{2}(q^{2m-2}+jq^{m-1})$} ~(j=1,-1),~Q_{0}=2q^{2m-3}, P_{0}=(2q-1)q^{2m-2}-2(q-1)q^{m-1}-1.$
\end{theorem}
\begin{center}
	\begin{tabular}[l]{p{8cm} p{6.6cm}}
		\text{\textbf{Table 1:} The presented code $\mathcal{C}_{D}$, for $m\geq3$ odd, has the following Lee-weight distribution }\\
		\hline
		Weight $w$ & Multiplicity $f$ \\  
		\hline  
		0 & 1  \\ 
		
		$2(q-1)q^{2m-3}$ & $(2q-1)q^{2m-2}-2(q-1)q^{m-1}-1 $ \\
		
		$2(q-1)(q^{2m-3}-q^{\frac{3m-5}{2}})$ 	& $(q-1)(q^{m-1}-q^{\frac{m-1}{2}})$\\
		$2(q-1)(q^{2m-3}+q^{\frac{3m-5}{2}})$ 	& $(q-1)(q^{m-1}+q^{\frac{m-1}{2}})$\\
		$2(q-1)(q^{2m-3}-q^{m-2})$	&$\frac{(q-1)^{2}}{2}(q^{2m-2}+q^{m-1}) $\\
		$2(q-1)(q^{2m-3}+q^{m-2})$	&$\frac{(q-1)^{2}}{2}(q^{2m-2}-q^{m-1}) $\\
		
		\hline
	\end{tabular}
\end{center}
\begin{proof}
	For odd  $m(\geq3)$, let $G=(-1)^{\frac{(q-1)^{2}m}{8}}q^{\frac{m}{2}}$ and $\overline{G}=(-1)^{\frac{(q-1)^{2}}{8}}q^{\frac{1}{2}}$.  By Lemma \ref{p3le7}, we obtain the length of the code $\mathcal{C}_{D}$=$N(0,0)-1=q^{2m-2}-1$. For $\textbf{c}\in\mathcal{C}_{D},$ we have  $wt_{H}(\phi(\textbf{c}))=(q-1)(\Omega_{\lambda}(\alpha,\beta)+\Omega_{\lambda}(\beta,\alpha))=2(q-1)\Omega_{\lambda}(\alpha,\beta)$ since $ \Omega(\alpha,\beta)=\Omega(\beta,\alpha).$\\
	
	Suppose that $\alpha=0$ and $\beta=0$. Then, by Lemmas \ref{p3le7} and \ref{p3le13}, we have $wt_{H}(\phi(\textbf{c}))=0.$
	This value occurs only once.\\
	Next, we suppose that $\alpha\neq0$ and $\beta\neq0$. Then, we have five distinct cases to consider:\\
	\textbf{Case 1:} If $\mathrm{Tr}(\alpha^{2})=0$ and $\mathrm{Tr}(\beta^{2})=0$ or $\mathrm{Tr}(\alpha^{2})=0$ and $\mathrm{Tr}(\beta^{2})\neq0$ or $\mathrm{Tr}(\alpha^{2})\neq0$ and $\mathrm{Tr}(\beta^{2})=0$ or $\mathrm{Tr}(\alpha^{2})=0$ or $\mathrm{Tr}(\beta^{2})=0$, then, from Lemmas \ref{p3le7} and \ref{p3le13}, we obtain
	$$
	wt_{H}(\phi(\textbf{c}))=2(q-1)q^{2m-3}.$$
	By Lemmas \ref{p3le4} and \ref{p3le7}, the frequency is  $(2q-1)q^{2m-2}-2(q-1)q^{m-1}-1$.\\
	\textbf{Case 2:} If $\overline{\eta}(\mathrm{Tr}(\alpha^{2}))=1$ or $\overline{\eta}(\mathrm{Tr}(\beta^{2}))=1$, then, by Lemmas \ref{p3le7} and \ref{p3le13}, one can obtain
	$$
	wt_{H}(\phi(\textbf{c}))=2(q-1)q^{2m-3}+2\overline{\eta}(-1)G\overline{G}(q-1)q^{m-3}.
	$$
	By Lemma \ref{p3le4}, the frequency is $(q-1)(q^{m-1}+q^{-1}\overline{\eta}(-1)G\overline{G})$.\\ 
	\textbf{Case 3:} If $\overline{\eta}(\mathrm{Tr}(\alpha^{2}))=-1$ or $\overline{\eta}(\mathrm{Tr}(\beta^{2}))=-1$, then, by Lemmas \ref{p3le7} and \ref{p3le13}, we have
	$$
	wt_{H}(\phi(\textbf{c}))=2(q-1)q^{2m-3}-2\overline{\eta}(-1)G\overline{G}(q-1)q^{m-3}.
	$$
	By Lemma \ref{p3le4}, the frequency is $(q-1)(q^{m-1}-q^{-1}\overline{\eta}(-1)G\overline{G})$.\\
	\textbf{Case 4:} If $\overline{\eta}(\mathrm{Tr}(\alpha^{2})\mathrm{Tr}(\beta^{2}))=1$, then, from Lemmas \ref{p3le7} and \ref{p3le13}, we have
	$$
	wt_{H}(\phi(\textbf{c}))=2(q-1)q^{2m-3}-2\overline{\eta}(-1)G^{2}(q-1)q^{-2}.
	$$
	By Lemma \ref{p3le7}, the frequency is $\frac{(q-1)^{2}}{2}(q^{2m-2}+q^{-1}\overline{\eta}(-1)G^{2})$.\\ 
	\textbf{Case 5:} If $\overline{\eta}(\mathrm{Tr}(\alpha^{2})\mathrm{Tr}(\beta^{2}))=-1$, then, from Lemmas \ref{p3le7} and \ref{p3le13}, one can acquire
	$$
	wt_{H}(\phi(\textbf{c}))=2(q-1)q^{2m-3}+2\overline{\eta}(-1)G^{2}(q-1)q^{-2}.
	$$
	By Lemma \ref{p3le7}, this value occurs  $\frac{(q-1)^{2}}{2}(q^{2m-2}-q^{-1}\overline{\eta}(-1)G^{2})$ times.\\
	Since, for $\textbf{c}\in\mathcal{C}_{D}$, $d_{L}(\textbf{c})$=$d_{H}(\phi(\textbf{c}))$. Thus, we obtain Table 1 which demonstrates Lee-weight distribution of $\mathcal{C}_{D}$.
	
	For $\lambda\in\mathbb{F}_{q}^{*}$ and $\phi(\textbf{c})=(y_{1},y_{2},\ldots,y_{2n})\in\phi(\mathcal{C}_{D}),$ denote
	$$\mathcal{A}_{\lambda}(\phi(\textbf{c}))=\#\{1\leq j\leq2n: y_{j}=\lambda\},$$ 
	where $n$ represents the length of the code $\mathcal{C}_{D}.$
	Then
	$$\mathcal{A}_{\lambda}(\phi(\textbf{c}))=2\Omega_{\lambda}(\alpha,\beta).$$ 
	Thus, by Lemmas \ref{p3le13} and \ref{p3le7}, we obtained the complete weight enumerator of $\phi(\mathcal{C}_{D}).$
	This concludes the theorem's proof.	
\end{proof}
\begin{example}
	Assume that $q=m=3$. Then the parameters of the presented linear code $\mathcal{C}_{D}$ and its corresponding linear code $\phi(\mathcal{C}_{D})$ are $[80, 3, 72]$ and $[160, 6, 72]$, respectively. Additionally, the Lee-weight enumerator of the linear code $\mathcal{C}_{D}$ and complete Hamming-weight  enumerator of the linear code $\phi(\mathcal{C}_{D})$ are $1+12 z^{72}+180z^{96}+368z^{108}+144z^{120}+24z^{144}$ and
	$$w_{0}^{160}+12w_{0}^{88}\prod_{i=1}^{2}w_{i}^{36}+180w_{0}^{64}\prod_{i=1}^{2}w_{i}^{48}+368w_{0}^{52}\prod_{i=1}^{2}w_{i}^{54}+144w_{0}^{40}\prod_{i=1}^{2}w_{i}^{60}+24w_{0}^{16}\prod_{i=1}^{2}w_{i}^{72},$$ respectively.
\end{example}
\begin{remark}\label{re1}	
	In comparison to the linear code $[160, 10, 48]$ in Example 3.1 of \cite{LL19}, the linear code $[160, 6, 72]$ in the preceding example has a better relative minimum distance.
	
\end{remark}
\begin{example}	Let $q=3~and~m=5$. Then the parameters of presented linear code $\mathcal{C}_{D}$ and its corresponding linear code $\phi(\mathcal{C}_{D})$ are $[6560, 5, 7776]$ and $[13120,10, 7776]$, respectively. Additionally, the Lee-weight enumerator of the linear code $\mathcal{C}_{D}$ and complete Hamming-weight  enumerator of its Gray image $\phi(\mathcal{C}_{D})$ are $1+144z^{7776}+13284z^{8640}+32480z^{8748}+12960z^{8856}+180z^{9720}$ and
	\begin{align*}
		&w_{0}^{13120}+144w_{0}^{5344}\prod_{i=1}^{2}w_{i}^{3888}+13284w_{0}^{4480}\prod_{i=1}^{2}w_{i}^{4320}+32480w_{0}^{4372}\prod_{i=1}^{2}w_{i}^{4374}\\&+12960w_{0}^{4264}\prod_{i=1}^{2}w_{i}^{4428}+180w_{0}^{3400}\prod_{i=1}^{2}w_{i}^{4860},
	\end{align*} respectively.	
\end{example}
\begin{theorem}\label{p3theorem2}
	Let $m\geq 2$ be even. The parameters for the presented  linear code $\mathcal{C}_{D}$ are $[(q^{m-1}+q^{-1}(q-1)G)^{2}-1,m]$ Additionally,  Table 2 contains its Lee-weight distribution, and the complete Hamming-weight enumerator of corresponding linear code $\phi(\mathcal{C}_{D})$ is	
	\begin{align*}
		&w_{0}^{2P_{3}^{2}-2}+2(P_{3}-1)w_{0}^{P_{2}+Q_{3}(q-1)(q+G)q^{-1}-2}\prod_{i=1}^{q-1}w_{i}^{P_{2}-Q_{3}q^{-1}G}\\&+2Q_{2}w_{0}^{P_{2}-Q_{3}q^{-1}G-2}\prod_{i=1}^{q-1}w_{i}^{P_{2}+Q_{3}+2Gq^{m-3}}\\&+(P_{3}-1)^{2}w_{0}^{P_{2}+(q-1)Q_{3}-2}\prod_{i=1}^{q-1}w_{i}^{P_{2}}+2Q_{2}(P_{3}-1)w_{0}^{P_{2}-2}\prod_{i=1}^{q-1}w_{i}^{P_{2}+Q_{3}}\\&+Q_{2}^{2}w_{0}^{P_{2}+Q_{3}-2}\prod_{i=1}^{q-1}w_{i}^{P_{2}+2(q-2)q^{m-2}},
	\end{align*}
	where $G=-(-1)^{\frac{(q-1)^{2}m}{8}}q^{\frac{m}{2}},~P_{2}=2q^{2m-3}+4G(q-1)p^{m-3},~\mbox{$Q_{2}=(p-1)(q^{m-1}-q^{-1}G)$},~\\\mbox{$P_{3}=q^{m-1}+q^{-1}(p-1)G$},~Q_{3}=2(q-1)q^{m-2}.$
\end{theorem}
\begin{center}
	\begin{tabularx}{\columnwidth}{p{9.8cm} p{5cm}}
		\text{\textbf{Table 2:} The presented code $\mathcal{C}_{D}$, for $m\geq2$ even, has the following Lee-weight distribution }\\
		\hline
		Weight w & Multiplicity $f$ \\  
		\hline  
		0 & 1  \\ 
		
		$2(q-1)q^{2m-3}+2G(q-1)^{2}q^{m-3}$ & $2(q^{m-1}+q^{-1}(q-1)G-1) $ \\
		
		$2(q-1)q^{2m-3}+2G(q-1)(2q-1)q^{m-3}+2(q-1)^{2}q^{m-2}$ 	& $2(q-1)(q^{m-1}-q^{-1}G)$\\
		$2(q-1)q^{2m-3}+4G(q-1)^{2}q^{m-3}$ 	& $(q^{m-1}+q^{-1}(q-1)G-1)^{2}$\\
		$2(q-1)q^{2m-3}+4G(q-1)^{2}q^{m-3}+2(q-1)^{2}q^{m-2}$	&$2(q-1)(q^{m-1}-q^{-1}G)(q^{m-1}+q^{-1}(q-1)G-1)$\\
		$2(q-1)q^{2m-3}+4G(q-1)^{2}q^{m-3}+2(q-1)(q-2)q^{m-2}$	&$(q-1)^{2}(q^{m-1}-q^{-1}G)^{2}$\\
		
		\hline
	\end{tabularx}
\end{center}
\begin{proof}
	For even $m(\geq2)$, let $G=-(-1)^{\frac{(q-1)^{2}m}{8}}q^{\frac{m}{2}}$. Then, from Lemma \ref{p3le7}, the length of the code $\mathcal{C}_{D}$= $N(0,0)-1=(q^{m-1}+q^{-1}(q-1)G)^{2}-1$. For $\textbf{c}\in\mathcal{C}_{D},$ we have  $wt_{H}(\phi(\textbf{c}))=(q-1)(\Omega_{\lambda}(\alpha,\beta)+\Omega_{\lambda}(\beta,\alpha))=2(q-1)\Omega_{\lambda}(\alpha,\beta)$ since $ \Omega(\alpha,\beta)=\Omega(\beta,\alpha).$\\
	Suppose that $\alpha=0$ and $\beta=0$, then we have, by Lemmas \ref{p3le7} and \ref{p3le13},
	$wt_{H}(\phi(\textbf{c}))=0.$
	This value only appears once.\\
	Next, we assume that $\alpha\neq0$ and $\beta\neq0$, then we have the following distinct cases to consider: \\
	\textbf{Case 1:} If $\mathrm{Tr}(\alpha^{2})=0$ or $\mathrm{Tr}(\beta^{2})=0$, then, by Lemmas \ref{p3le7} and \ref{p3le13}, we have
	$$
	wt_{H}(\phi(\textbf{c})=2(q-1)q^{2m-3}+2G(q-1)^{2}q^{m-3}.
	$$
	By Lemmas \ref{p3le4}, the frequency is  $2(q^{m-1}+q^{-1}(q-1)G-1)$.\\
	\textbf{Case 2:} If $\mathrm{Tr}(\alpha^{2})\neq0$ or $\mathrm{Tr}(\beta^{2})\neq0$, then, from Lemmas \ref{p3le7} and \ref{p3le13}, we obtain
	$$
	wt_{H}(\phi(\textbf{c}))=2(q-1)q^{2m-3}+2G(q-1)(2q-1)q^{m-3}+2(q-1)^{2}q^{m-2}.
	$$
	By Lemma \ref{p3le4}, the frequency is $2(q-1)(q^{m-1}-q^{-1}G)$.\\
	\textbf{Case 3:} If $\mathrm{Tr}(\alpha^{2})=0$ and $\mathrm{Tr}(\beta^{2})=0$, then, from Lemmas \ref{p3le7} and \ref{p3le13}, we have
	$$
	wt_{H}(\phi(\textbf{c}))=2(q-1)q^{2m-3}+4G(q-1)^{2}q^{m-3}.
	$$
	By Lemmas \ref{p3le4} and \ref{p3le7}, the frequency is $(q^{m-1}+q^{-1}(q-1)G-1)^{2}$.\\
	\textbf{Case 4:} If $\mathrm{Tr}(\alpha^{2})=0\text{ and } \mathrm{Tr}(\beta^{2})\neq0$ or $\mathrm{Tr}(\alpha^{2})\neq0\text{ and } \mathrm{Tr}(\beta^{2})=0$, then, by Lemmas \ref{p3le7} and \ref{p3le13}, one can obtain
	$$
	wt_{H}(\phi(\textbf{c}))=2(q-1)q^{2m-3}+4G(q-1)^{2}q^{m-3}+2(q-1)^{2}q^{m-2}.
	$$
	By Lemmas \ref{p3le4} and \ref{p3le7}, the frequency is $2(q-1)(q^{m-1}-q^{-1}G)(q^{m-1}+q^{-1}(q-1)G-1)$.\\
	\textbf{Case 5:} If $\mathrm{Tr}(\alpha^{2})\neq0\text{ and }\mathrm{Tr}(\beta^{2})\neq0$, then,
	$$
	wt_{H}(\phi(\textbf{c}))=2(q-1)q^{2m-3}+4G(q-1)^{2}q^{m-3}+2(q-1)(q-2)q^{m-2})
	$$
	follows from Lemmas \ref{p3le7} and \ref{p3le13}.
	By Lemma \ref{p3le7}, this value occurs  $(q-1)^{2}(q^{m-1}-q^{-1}G)^{2}$ times.\\
	Since, for $\textbf{c}\in\mathcal{C}_{D}$, $d_{L}(\textbf{c})$=$d_{H}(\phi(\textbf{c}))$. Thus, we have the Lee-weight distribution of $\mathcal{C}_{D}$ which completes the Table 2. For $\lambda\in\mathbb{F}_{q}^{*}$ and $\phi(\textbf{c})=(y_{1},y_{2},\ldots,y_{2n})\in\phi(\mathcal{C}_{D}),$ denote
	$$\mathcal{A}_{\lambda}(\phi(\textbf{c}))=\#\{1\leq j\leq2n: y_{j}=\lambda\},$$
	where $n$ stands for the length of the code $\mathcal{C}_{D}.$ Then
	$$\mathcal{A}_{\lambda}(\phi(\textbf{c}))=2\Omega_{\lambda}(\alpha,\beta).$$ 
	Thus, by Lemmas \ref{p3le7} and \ref{p3le13}, we obtain the complete weight enumerator of $\phi(\mathcal{C}_{D}).$ This completes the proof. 	
\end{proof}
\begin{example}
	Assume that $q=3~and~m=4$. Then the parameters of presented linear code $\mathcal{C}_{D}$ and its corresponding linear code $\phi(\mathcal{C}_{D})$ are $[440, 4, 504]$ and $[880, 8, 504]$, respectively. Additionally, the Lee-weight enumerator of the linear code $\mathcal{C}_{D}$ and complete Hamming-weight  enumerator of the linear code $\phi(\mathcal{C}_{D})$ are $1+120z^{504}+400z^{540}+3600z^{576}+2400z^{612}+40z^{756}$ and
	\begin{align*}
		&w_{0}^{880}+120w_{0}^{376}\prod_{i=1}^{2}w_{i}^{252}+400w_{0}^{340}\prod_{i=1}^{2}w_{i}^{270}+3600w_{0}^{304}\prod_{i=1}^{2}w_{i}^{288}+2400w_{0}^{268}\prod_{i=1}^{2}w_{i}^{306}\\&+40w_{0}^{124}\prod_{i=1}^{2}w_{i}^{378},
	\end{align*}
	respectively.
\end{example}

\section{Application to secret sharing schemes}\label{Sec.5}
The set $s(x)=\{i:x_{i}\neq0,1\leq i\leq n\}$  is referred to as the support of $x$, where  $x=(x_{1},x_{2},\ldots,x_{n})\in\mathbb{F}_{q}^{n}$.  We say that vector $x$ covers vector $y$ if $ s(y)\subseteq s(x)$. If there is not another nonzero codeword $y$ such that $s(y)\subseteq s(x)$ exists, then a nonzero codeword $x$ of a linear code $\mathcal{C}$ is said to be minimal. Identifying minimal codewords for a given linear code is a challenging problem. However, a numerical condition related to the weights of the code, discussed in \cite{AB98}, provides insight into whether codewords in a linear code  $\mathcal{C}$ are minimal. In secret sharing schemes with interesting access structure, minimal codewords have uses (see \cite{YD06}).
\begin{lemma}\label{ABlemma}
	(Ashikhmin-Barg)\cite{AB98}
	Denote the minimum and maximum nonzero weights by $w_{0}$ and $w_{\infty}$, respectively. Every nonzero codeword of $\mathcal{C}$ is minimal if the following condition holds: $$\frac{w_{0}}{w_{\infty}}>\frac{(q-1)}{q}.$$
	\end{lemma}
\begin{proposition}
For odd $m(\geq5)$, the nonzero codewords of $q$-ary linear code $\phi(\mathcal{C}_{D})$ are minimal.
\end{proposition}
\begin{proof} Let $m\geq5$ be odd. Then, for the Gray image $\phi(\mathcal{C}_{D})$ of the presented code $\mathcal{C}_{D}$ in Theorem \ref{p3theorem1}, we get
$$\frac{w_{0}}{w_{\infty}}=\frac{2(q-1)(q^{2m-3}-q^{\frac{3m-5}{2}})}{2(q-1)(q^{2m-3}+q^{\frac{3m-5}{2}})}=\frac{(q^{2m-3}-q^{\frac{3m-5}{2}})}{(q^{2m-3}+q^{\frac{3m-5}{2}})}.$$
From the fact that $2q^{\frac{3m-3}{2}}<q^{2m-3}+q^{\frac{3m-5}{2}}$ when $m\geq5$ and $q$ is an odd prime, we can conclude that
$$\frac{w_{0}}{w_{\infty}}>\frac{(q-1)}{q}.$$This  completes the proof.\end{proof}
\begin{proposition}
For even $m(\geq6)$, the nonzero codewords of $q$-ary linear code $\phi(\mathcal{C}_{D})$ are minimal.
\end{proposition}
\begin{proof}
Assume that $m\geq4$ is even. Then, for the Gray image $\phi(\mathcal{C}_{D})$ of the presented code $\mathcal{C}_{D}$ in Theorem \ref{p3theorem2}, we get
$$\frac{w_{0}}{w_{\infty}}=\frac{2(q-1)q^{2m-3}+2(q-1)^{2}q^{\frac{3m-6}{2}}}{2(q-1)q^{2m-3}+2(q-1)(2q-1)q^{\frac{3m-6}{2}}+2(q-1)^{2}q^{m-2}}$$  or 
$$\frac{w_{0}}{w_{\infty}}=\frac{2(q-1)q^{2m-3}-2(q-1)(2q-1)q^{\frac{3m-6}{2}}+2(q-1)^{2}q^{m-2}}{2(q-1)q^{2m-3}-2(q-1)^{2}q^{\frac{3m-6}{2}}}.$$ 
Then it can easily be seen that 
$$\frac{q^{2m-3}+(q-1)q^{\frac{3m-6}{2}}}{q^{2m-3}+(2q-1)q^{\frac{3m-6}{2}}+(q-1)q^{m-2}}>\frac{q-1}{q},\text{ for every $m\geq4$}$$ and 
$$\frac{q^{2m-3}-(2q-1)q^{\frac{3m-6}{2}}+(q-1)q^{m-2}}{q^{2m-3}-(q-1)q^{\frac{3m-6}{2}}}>\frac{q-1}{q},\text{ for every $m\geq6$}.$$ 
So, we can conclude that $\frac{w_{0}}{w_{\infty}}>\frac{(q-1)}{q}$ for every $m\geq6$, which completes the proof.\end{proof}
\section{Conclusions}\label{Sec.6}
In this paper, we have introduced a family of five-Lee-weight linear codes constructed using a specific defining set. Also, we have determined the complete Hamming-weight enumerator of corresponding linear code $\phi(\mathcal{C}_{D})$. In Section \ref{Sec.5}, we have shown that the constructed codes are minimal under some restriction on variables. Minimal codes are particularly valuable for constructing secret sharing schemes. The members of defining set used to construct codes in \cite{LL19}, belong to  $\mathbb{F}_{q^{m}}$ but the members of defining set used in the present paper belong to  $\mathbb{F}_{q^{m}}+u\mathbb{F}_{q^{m}}.$ Moreover, the ring that we have considered in this paper is more general than that in \cite{LL19}. More useful codes can be constructed over such rings based on defining sets and we leave this for future research work.

\textbf{Acknowledgements.} This research work was done with the affiliation Aligarh Muslim University, and was supported by the University Grants Commission, New Delhi, India, under  JRF in Science, Humanities $\&$ Social Sciences scheme with Grant number  11-04-2016-413564.

\textbf{Conflict of interest} The authors have mutually agreed on the order of authorship, acknowledging that both have contributed equally to the completion of the paper. They also confirm that appropriate acknowledgments have been made to all funding organizations that supported their work

\end{document}